%% file: main.tex
\documentclass[a4paper,11pt]{article}

\input{preamble}

\usepackage{marvosym}

\title{Hamiltonicity Parameterized by Mim-Width is (Indeed) Para-NP-Hard\thanks{A preliminary version of this work appears in the proceedings of IPEC 2025.}}

\author[1]{Benjamin Bergougnoux}
\author[2,\Letter]{Lars Jaffke}

\affil[1]{Aix-Marseille Université, LIS, CNRS, France}
\affil[2]{NHH Norwegian School of Economics, Bergen, Norway}
\affil[\Letter]{Corresponding author}

\date{\today}

\begin{document}

\maketitle

\input{abstract}

\begin{center}
    \begin{minipage}{.875\textwidth}
        \small
        \noindent%
        \textbf{Keywords.} Hamiltonian Path, Hamiltonian Cycle, Mim-Width, Para-NP-Hardness
    \end{minipage}
\end{center}

\input{content}

\myparagraph{Acknowledgements.}
We thank Richard Wols for the counterexample presented in \Cref{sec:counterex}.

\bibliographystyle{plain}
\small
\bibliography{ref}

\end{document}

%% file: preamble.tex
\usepackage{authblk}

\usepackage{graphicx} 
\usepackage[margin=3cm]{geometry}
\usepackage{xspace}

\usepackage{amsmath}
\usepackage{amsthm}
\usepackage{amssymb}

\newtheorem{lemma}{Lemma}

\usepackage{thmtools}

\usepackage{cleveref}

\theoremstyle{definition}
\newtheorem{definition}{Definition}
\newtheorem{observation}{Observation}

\let\plainqed\qedsymbol

\usepackage{macros}

\Crefname{observation}{Observation}{Observations}

%% file: abstract.tex
\newcommand\mimbound{26}

\begin{abstract}
    We prove that \textsc{Hamiltonian Path} and \textsc{Hamiltonian Cycle} 
    are \NP-hard on graphs of linear mim-width $\mimbound$,
    even when a linear order of the input graph with mim-width $\mimbound$ 
    is provided together with input.
    This fills a gap left by a broken proof of the para-NP-hardness of 
    Hamiltonicity problems parameterized by mim-width.
\end{abstract}

%% file: content.tex
\section{Introduction}
%
%
Maximum induced matching width (mim-width)
which was introduced in 2012 by Vatshelle~\cite{VatshelleThesis}
is by now a common staple among width parameters of graphs (e.g.~\cite{BakkaneJ22,BelmonteV13,BBD25,BergougnouxDJ23,BHJL25,BergougnouxK21,BergougnouxPT22,BodlaenderGJJL25,BonomoBrabermanBMP24,BrettellHMPP22,BrettellHMP22,BrettellMPY25,Bui-XuanTV13,EibenGHJK22,GalbyMR20,GolovachHKKSV18,GonzalezM24,HogemoTV19,JaffkeKST19,JKT20,JaffkeKT20,JaffkeKT24,Mengel18,MunaroY23,OtachiST24}).
One pillar in the understanding of its algorithmic use and limitations 
is that finding simply structured \emph{induced} subgraphs 
tends to be tractable (solvable in \XP time by the width of a given decomposition),
while finding \emph{non-induced} subgraphs is usually hard (\NP-hard for constant mim-width).
Most notably, finding long induced paths and cycles is in \XP~\cite{BergougnouxK21},
while in \cite{JKT20}, it has been claimed that 
\textsc{Hamiltonian Cycle} is \NP-hard on graphs of linear mim-width~$1$.
The latter claim rests on a proof that \textsc{Hamiltonian Cycle} is \NP-hard on 
rooted directed path graphs given in \cite{PP08},
and in \cite{JKT20} it is shown that the graphs resulting from this reduction 
have linear mim-width $1$.
Unfortunately, as we show below, the reduction from \cite{PP08} 
has a flaw, 
which means that the complexity of Hamiltonicity problems 
(or, more generally, problems of finding long non-induced paths or cycles)
parameterized by mim-width remained open.
In this paper, we repair this gap in the literature 
by proving the following result.
\begin{restatable}{theorem}{thmMain}
    The \textsc{Hamiltonian Path} and \textsc{Hamiltonian Cycle} problems are \NP-hard on graphs of linear mim-width \mimbound,
    even when a linear order of the input graph of mim-width $\mimbound$ is provided as part of the input.
\end{restatable}

Observe that hardness holds even when 
a linear order of small mim-width is provided at the input.
Recently, Bergougnoux, Bonnet, and Duron~\cite{BBD25} showed that 
computing the mim-width of a graph 
exactly is para-\NP-hard as well.
Note however that there might still be some $f$ 
such that there is an \XP time $f(k)$-approximation algorithm to mim-width~$k$.

We would like to mention that 
the constant on the mim-width in our statements can likely be improved,
and it is an interesting question where exactly the boundary 
between tractability and intractability lies for Hamiltonicity problems.
Our results very much allow the possibility 
that Hamiltonicity is polynomial-time solvable 
on rooted directed path graphs 
or more generally, on graphs of mim-width $1$.
For \textsc{Clique} and a number of related problems, 
Otachi, Suzuki, and Tamura~\cite{OtachiST24} recently showed 
that the boundary between tractable and intractable cases 
lies between mim-width $1$ and $2$.
While \textsc{Clique} was known to be solvable in polynomial time 
on graphs of mim-width $1$ (given a decomposition),
they showed that it becomes \NP-hard on graphs of mim-width~$2$.
It would be interesting to establish a similar dividing line for Hamiltonicity problems.

\section{Preliminaries}
We denote the set of natural numbers by $\bN = \{0, 1, \ldots\}$,
and for $n \in \bN$, we use the shorthand $[n] = \{1, \ldots, n\}$.
All graphs considered in this paper are finite, undirected, and simple.
For a graph $G$, 
we denote by $V(G)$ its vertices and by $E(G)$ its edges.
For a vertex $v \in V(G)$, we denote its neighborhood by 
$N_G(v) = \{u \mid uv \in E(G)\}$.
The \emph{degree} of $v$ is $\card{N_G(v)}$.
For a set $X \subseteq V(G)$, 
the \emph{subgraph of $G$ induced by $X$}, denoted by $G[X]$,
is the graph on vertex set $X$ and edge set 
$\{uv \mid uv \in E(G), \{u, v\} \subseteq X\}$.
A set $X \subseteq V(G)$ is \emph{independent} if $G[X]$ has no edges 
and $X$ is a \emph{clique} if every pair of vertices in $X$ is adjacent in $G$.

A \emph{walk} in a graph $G$ is a sequence of vertices 
$W = (v_1, \ldots, v_r)$ such that for all $i \in [r-1]$, $v_i v_{i+1} \in E(G)$.
If all vertices in $W$ are pairwise distinct, then $W$ is called a \emph{path},
and if all vertices in $\{v_1, \ldots, v_{r-1}\}$ are pairwise distinct and $v_1 = v_r$,
then $W$ is a \emph{cycle}.
A path/cycle is \emph{Hamiltonian} if it contains all vertices of $G$.

A graph $G$ is \emph{bipartite} if there is a partition $(A, B)$ of $V(G)$
such that $A$ and $B$ are independent.
$G$ is further called a \emph{chain graph}
if there is a linear order $a_1, \ldots, a_n$ of $A$
such that for all $i \in [n-1]$: $N(a_i) \subseteq N(a_{i+1})$.

Let $G$ be a graph. Two distinct vertices $u, v \in V(G)$ are \emph{false twins} if $N_G(u) = N_G(v)$ and $uv \notin E(G)$.
Let $uv \in E(G)$. 
The \emph{subdivision of $uv$} is the operation of replacing $uv$ 
with a path $(u, x, v)$ where $x$ is a newly created vertex.

\myparagraph{Linear Mim-width.}
Let $G$ be a graph and $A \subseteq V(G)$.
We let $\bar{A} = V(G) \setminus A$.
We denote by $G[A, \bar{A}]$ 
the bipartite subgraph of $G$ with vertex bipartition $(A, \bar{A})$
whose edges are $\{ab \mid ab \in E(G), a \in A, b \in \bar{A}\}$.
A set of edges $M = \{a_i b_i \mid i \in [r]\} \subseteq E(G)$ 
is a \emph{matching} if all vertices in 
$V(M) = \bigcup_{i \in [r]} \{a_i, b_i\}$ 
are pairwise distinct.
A matching $M$ is \emph{induced} if $E(G[V(M)]) = M$.
The \emph{mim-value} of a set $A \subseteq V(G)$, denoted by $\mimval_G(A)$, 
is the largest size of any induced matching in $G[A, \bar{A}]$.
Let $\linord = (v_1, \ldots, v_n)$ be a linear order of $V(G)$.
The \emph{maximum induced matching width (mim-width) of $\linord$}
is $\max_{i \in [n]} \mimval_G(\{v_1, \ldots, v_i\})$.

\medskip
The following observation will be helpful in a later proof.
\begin{observation}\label{obs:chain}
    Let $H$ be a chain graph with bipartition $(A, B)$.
    Then, $\mimval_H(A) \le 1$.
\end{observation}

\subsection{Counterexample to the existing reduction}\label{sec:counterex}
We recall the reduction from~\cite{PP08}.
It is from \textsc{Hamiltonian Cycle} on bipartite graphs of maximum degree $3$ and minimum degree $2$
to \textsc{Hamiltonian Cycle} on rooted directed path graphs,
which are intersection graphs of directed paths in a rooted tree.

Let $G$ be a bipartite graph of maximum degree $3$ and minimum degree $2$, 
and let $(A, B)$ be the 
vertex bipartition of $G$.
We may assume that $\card{A} = \card{B} = n$,
and fix arbitrary orderings on $A$ and $B$,
that is, we let 
$A = \{a_1, \ldots, a_n\}$ and 
$B = \{b_1, \ldots, b_n\}$ throughout.
We construct a graph $H$ as follows.
\begin{enumerate}
    \item Let $H$ be a copy of $G$.
    \item Subdivide each edge $a_ib_j \in E(H)$ once. For convenience, we use $a_i b_j$ to both denote the edge in $G$ and the vertex in $H$ created in the subdivision of this edge.
    \item Make $\{a_ib_j \mid i, j \in [n], a_ib_j \in E(G)\}$ a clique in $H$.
    \item For each $i \in [n]$, make $a_i$ adjacent to all vertices $a_{i'}b_j \in V(H)$ where $1 \le i' \le i$, $j \in [n]$, $a_{i'}b_j \in E(G)$.
    \item For each $i \in [n]$ such that $b_i$ has degree $3$ in $G$, 
        add the vertex $b_i'$ with $N_H(b_i') = N_H(b_i)$ to $H$.
\end{enumerate}
In \cite{JKT20} it is shown that the graph $H$ as constructed above has linear mim-width~$1$.

In \Cref{fig:counterexample},
we show a bipartite graph $G$ of maximum degree $3$ and minimum degree $2$ that has no Hamiltonian cycle, 
while the graph obtained from running the above reduction on input $G$ has one, 
thus providing a counterexample to the proof in~\cite{PP08}.
\begin{figure}[htb]
    \centering
    \includegraphics[width=\linewidth]{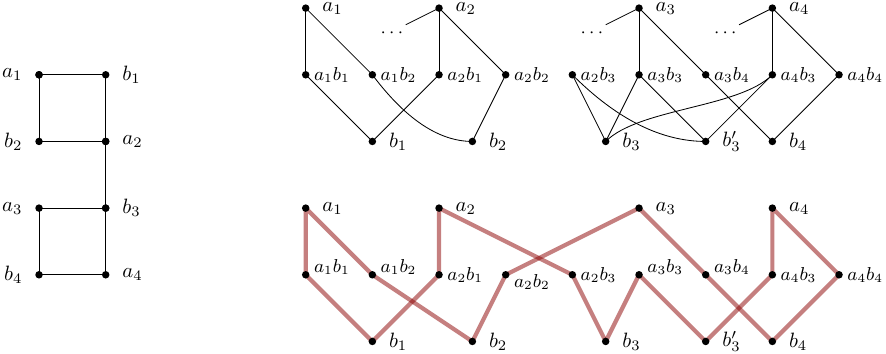}
    \caption{Counterexample to the construction from \cite{PP08}. 
    The three dots symbolize adjacency to all vertices to the left on the middle layer.}
    \label{fig:counterexample}
\end{figure}

\subsection{\ThreeSat}
A \emph{$3$-CNF formula} is a Boolean formula in conjunctive normal form such that each clause has at most three pairwise distinct literals.
%
Given a formula $\cnf$, we denote by $\var(\cnf)$ its set of variables. A truth assignment to the variables of $\cnf$ is a a function $\alpha \colon \var(\cnf)\to \{0,1\}$. We say that a clause $C$ is satisfied by $\alpha$ if $C$ contains a literal $\ell$ 
with variable $x$ such that $\alpha(x) = 1$ if $\ell = x$ and $\alpha(x) = 0$ if $\ell = \bar{x}$.
Moreover, $\alpha$ is said to be a satisfying assignment of $\cnf$ if it satisfies all the clauses of $\cnf$.
We will reduce from the following problem.
\problemStatement{\ThreeSat}{
    Input={A $3$-CNF formula $\phi$.},
    Question={Does $\cnf$ admit a satisfying assignment?}
}

\section{Main result}
In this section, we prove that the following problem is \NP-hard for 
parameter value $w = 25$.
\problemStatement{Hamiltonian Path/Linear Mim-Width}{
    Input={A graph $G$ and a linear order $\linord$ of $V(G)$.},
    Parameter={The mim-width $w$ of $\linord$.},
    Question={Does $G$ have a Hamiltonian path?}
}

We give a polynomial-time reduction from the \ThreeSat problem.
Let $\cnf$ be a $3$-CNF formula with 
variables $\var(\cnf) = \{x_1, \ldots, x_n\}$ and 
clauses $C_1, \ldots, C_m$.

\myparagraph{Variable gadgets.}
For each $i \in [n]$, we create a variable gadget as follows.
It consists of a sequence of $m$ cycles on six vertices such that 
traversing each cycle clockwise 
corresponds to setting $x_i$ to true, 
and traversing them counterclockwise corresponds to setting $x_i$ to false.
The construction ensures that for each variable, 
all cycles are traversed in the same way.

Let $i \in [n]$.
We show how to construct the variable gadget $\vargadget_i$.
For each $j \in [m]$, it contains a cycle $D_i^j$
whose vertices (listed in order) are called
$\varvertex{0}{in}{i}{j}$,
$\varvertex{0}{out}{i}{j}$,
$\varvertexdt{out}{i}{j}$,
$\varvertex{1}{out}{i}{j}$,
$\varvertex{1}{in}{i}{j}$,
$\varvertexdt{in}{i}{j}$.
For each $j \in [m-1]$,
we add an edge between
$\varvertex{1}{out}{i}{j}$ and $\varvertex{1}{in}{i}{j+1}$,
and between
$\varvertex{0}{out}{i}{j}$ and $\varvertex{0}{in}{i}{j+1}$.
See \Cref{fig:variable-gadget} for an illustration.
Moreover, there is a vertex $s_i$, 
adjacent to
$\varvertex{0}{in}{i}{1}$ and $\varvertex{1}{in}{i}{1}$,
and a vertex $t_i$, 
adjacent to 
$\varvertex{0}{out}{i}{m}$ and $\varvertex{1}{out}{i}{m}$.
\begin{figure}[hbt]
    \centering
    \includegraphics[width=\linewidth]{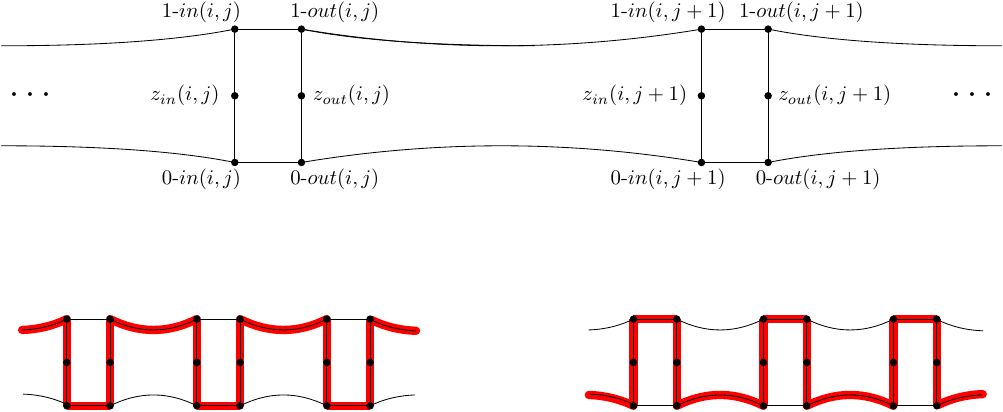}
    \caption{Illustration of the variable gadget. 
        Bottom left (resp., right): a traversal of the sequence corresponding to setting the variable to false (resp., true).}
    \label{fig:variable-gadget}
\end{figure}

\newcommand\gadget{\Gamma}
\myparagraph{Clause gadgets.}
For each clause $C_j$, we add a gadget $\clausegadget_j$ with the following property.
Say variable $x_i$ appears in $C_j$.
If a given truth assignment to $x_i$ satisfies $C_j$,
then we can enter and collect the vertices of $\clausegadget_j$
from $D_i^j$ (in the variable gadget $\vargadget_i$) 
under the traversal of $D_i^j$ corresponding to that truth assignment to $x_i$.
To avoid ``cheating'', we need $\clausegadget_j$ to have the following property:
once a path enters a vertex of $\clausegadget_j$,
then we have to immediately collect all vertices from $\clausegadget_j$,
and leave again via a prescribed exit, depending on where we entered $\clausegadget_j$.
The following clause gadget due to Cygan, Kratsch, and Nederlof~\cite{CKN18}
achieves just that.
We visualize these graphs in \Cref{fig:clause-gadget}.
\begin{figure}[htb]
    \centering
    \includegraphics[width=.7\linewidth]{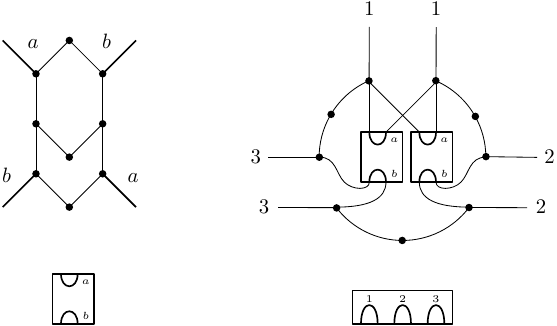}
    \caption{The clause gadgets due to Cygan, Kratsch, and Nederlof~\cite{CKN18}. On the left, a gadget such that each Hamiltonian path that enters it via an edge labeled $a$ (resp., $b$) must immediately collect all vertices in the gadget and leave via the other edge labeled $a$ (resp, $b$). 
    This gadget can be used for clauses of size two.
    On the right, a gadget for a clause of size three, with three entry-exit pairs and the analogous functionality.}
    \label{fig:clause-gadget}
\end{figure}
\newcommand\distin{\sigma}
\newcommand\distout{\tau}
\begin{lemma}[Cygan, Kratsch, and Nederlof~\cite{CKN18}]\label{lem:clause-gadget}
    For each $k \in \{2, 3\}$, there is a graph $\gadget_k$ on at most $27$ vertices
    with $k$ distinguished pairs of unique vertices 
    $\{(\distin_1, \distout_1), \ldots, (\distin_k, \distout_k)\}$
    with the following property.
    Let $G$ be a graph with an induced subgraph $\gadget_k$, 
    and let $P$ be a Hamiltonian path of $G$.
    Then, $P \cap \gadget_k$ is a $(\distin_i, \distout_i)$-path for some $i \in [k]$.
\end{lemma}

\myparagraph{Main construction.}
We now construct the graph $G$ from the formula $\cnf$,
see \Cref{fig:construction} for an illustration.
\begin{figure}[bth]
    \centering
    \includegraphics[width=.95\linewidth, page=1]{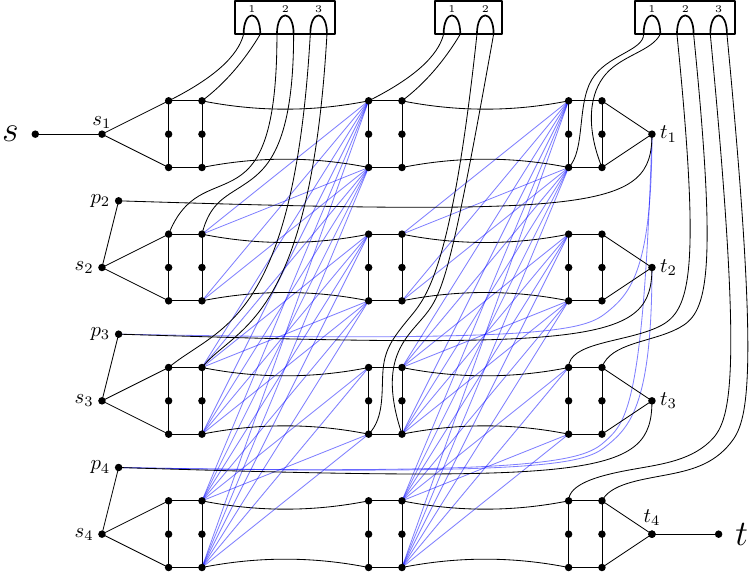}
    \caption{Overview of the construction. This is the example for the formula $(x_1 \lor x_2 \lor x_3) \land (x_1 \lor \overline{x_3}) \land (\overline{x_1} \lor x_3 \lor x_4)$.  The dummy edges are depicted in blue.}
    \label{fig:construction}
\end{figure}
First, we add two vertices $s$ and $t$ to $G$,
which will have degree one and are therefore the endpoints of 
any Hamiltonian path in $G$, if one exists.
For each $i \in [n]$, 
we add a variable gadget $\vargadget_i$.
We furthermore add the edges $s s_1$, $t_n t$, 
and for all $i \in [n-1]$, we proceed as follows.
We add a vertex $p_{i+1}$ to $G$, 
and add the edges $t_i p_{i+1}$ and $p_{i+1} s_{i+1}$.

Next, for each clause $C_j$ with literals $\ell_1, \ldots, \ell_k$,
we proceed as follows.
We let the clause gadget $\clausegadget_j$ be a copy of 
$\gadget_k$ from \Cref{lem:clause-gadget};
let $(\distin_1, \distout_1)$, $\ldots$, $(\distin_k, \distout_k)$ be its distinguished vertices.
For each $h \in [k]$, let $x_{j_h}$ be the variable of literal $\ell_h$.
If $x_{j_h}$ appears positively in $C_j$, then
we add the edges $\{\distin_h, \varvertex{0}{in}{i}{j}\}$ and $\{\distout_h, \varvertex{0}{out}{i}{j}\}$,
and if $x_{j_h}$ appears negated in $C_j$, then we add
$\{\distin_h, \varvertex{1}{in}{i}{j}\}$ and $\{\distout_h, \varvertex{1}{out}{i}{j}\}$.
This finishes the main part of the construction.

\myparagraph{Dummy edges.} 
At this point, $G$ may have arbitrarily high mim-width.
To reduce the mim-width of the graph, we add lots of dummy edges.
These edges are introduced in such a way that they decrease the mim-width
without allowing to ``cheat''.
That is, no Hamiltonian path of $G$ will ever use any of these edges.
We proceed as follows.
For each $j \in [m-1]$,
each $i \in [2..n]$, 
and each $h < i$,
we add the edges
$\{\varvertex{b_1}{out}{i}{j}, \varvertex{b_2}{in}{h}{j+1}\}$,
for any $b_1, b_2 \in \{0, 1\}$.
Moreover, for each $i,j\in [t]$ with $i<j$, we add the edge $\{t_i, p_j\}$.

\medskip
This finishes the construction.
Clearly, this reduction can be performed in polynomial time.
We show that it is correct.
\begin{lemma}\label{lem:cor-ass-to-path}
    If $\cnf$ is satisfiable, then $G$ has a Hamiltonian path.
\end{lemma}
\begin{proof}
    Suppose that $\cnf$ has a satisfying assignment $\ass \colon \var(\cnf) \to \{0, 1\}$.
    We construct a Hamiltonian path $P$ of $G$ as follows. See \Cref{fig:hamiltonian} for an illustration of this construction.
    First, $P$ starts in $s$, $s_1$.
    If $\ass(x_1) = 1$, then the next vertex of $P$ is 
    $\varvertex{1}{in}{1}{1}$, 
    and if $\ass(x_1) = 0$, then the next vertex of $P$ is
    $\varvertex{0}{in}{1}{1}$.
    Suppose the former is the case, the latter case is symmetric.
    The next vertices on $P$ are 
    $\varvertexdt{in}{1}{1}$,
    $\varvertex{0}{in}{1}{1}$.
    Now, if $x_1$ appears positively in $C_1$,
    then by construction, we can enter the gadget $\clausegadget_1$
    from $\varvertex{0}{in}{1}{1}$,
    collect all its vertices, and leave it to arrive at
    $\varvertex{0}{out}{1}{1}$.
    Note that in this case, $x_1$'s truth assignment satisfied $C_1$.
    $P$ then collects 
    $\varvertexdt{out}{1}{1}$,
    $\varvertex{1}{out}{1}{1}$,
    and takes the edge to
    $\varvertex{1}{in}{1}{2}$.
    We proceed analogously in sequence on all $D_1^j$,
    until we reach $t_1$.
    Here we follow the edge to $p_2$, $s_2$,
    and proceed on $\vargadget_2$ as we did on $\vargadget_1$,
    now based on which value $\ass$ assigns $x_2$.
    We repeat this process until we finally reach $t$.

    Whenever the truth assignment $\ass(x_i)$ satisfies a clause $C_j$,
    we are able to collect the vertices of $\clausegadget_j$.
    Therefore, the procedure described above indeed produces a Hamiltonian path of $G$.
\end{proof}

\begin{figure}
    \centering
    \includegraphics[width=\linewidth, page=2]{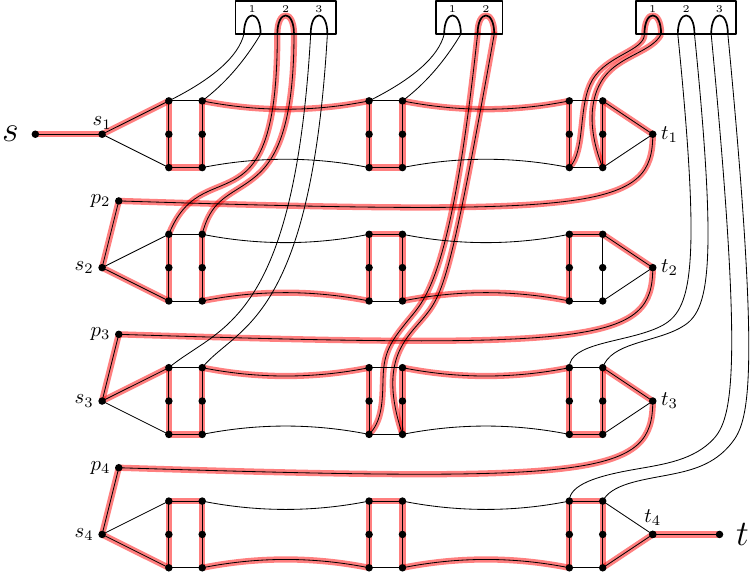}
    \caption{Example of the Hamiltonian path---of the graph presented in \Cref{fig:construction}---constructed in the proof of \Cref{lem:cor-ass-to-path} from the satisfying assignment $\alpha$ with $\alpha(x_1)=\alpha(x_3)= 0$ and $\alpha(x_2)=\alpha(x_4)=1$. We omit the dummy edges to improve the legibility. }
    \label{fig:hamiltonian}
\end{figure}

Let $s, t \in V(G)$, and let $P$ be an $(s, t)$-path in $G$.
For distinct $u, v \in V(P)$, 
we write $u \prec_P v$ if the path from $s$ to $v$ contains $u$.
For two sets $X, Y \subseteq V(G)$,
we write $X \prec_P Y$ if for all $x \in X$, $y \in Y$, $x \prec_P y$.
For intuition on the following notions, recall \Cref{fig:variable-gadget}.
\begin{definition}
    Let $a \in [n]$ and $b \in [m]$,
    and let $P$ be a Hamiltonian path of $G$ starting in~$s$.
    We say that $P$ \emph{traverses $D_a^1,\dots,D_a^b$ in true-order} 
    if $P$ contains the following as subpaths for every $j\in[b]$
    \begin{description} 
        \item[($P_{in}$)] $\varvertex{1}{in}{a}{j}$, $\varvertexdt{in}{a}{j}$, $\varvertex{0}{in}{a}{j}$ 
        \item[($P_{out}$)] $\varvertex{0}{out}{a}{j}$, $\varvertexdt{out}{a}{j}$, $\varvertex{1}{out}{a}{j}$
    \end{description}
    and for every $j\in [b-1]$, $P$ contains the edge $\{\varvertex{1}{out}{i}{j}, \varvertex{1}{in}{i}{j+1}\}$.
    We say that $P$ \emph{traverses $D_a^1,\dots,D_a^b$ in false-order} 
    if $P$ contains the following as subpaths: 
    \begin{description}
        \item[($P_{in}$)] $\varvertex{0}{in}{a}{j}$, $\varvertexdt{in}{a}{j}$, $\varvertex{1}{in}{a}{j}$
        \item[($P_{out}$)] $\varvertex{1}{out}{a}{j}$, $\varvertexdt{out}{a}{j}$, $\varvertex{0}{out}{a}{j}$
    \end{description}
    and for every $j\in [b-1]$, $P$ contains the edge $\{\varvertex{0}{out}{i}{j}, \varvertex{0}{in}{i}{j+1}\}$.
\end{definition}

\begin{definition}\label{def:respect}
    Let $a \in [n], b\in [m]$ and let $P$ be a Hamiltonian $(s, t)$-path of $G$.
    We say that $P$ is $(a,b)$-\emph{respectable} if the following are satisfied.
    \begin{enumerate}
        \item For all $i\in [n]$ with $i < a$, we have $V(D_i^{1}) \prec_P V(D_i^2) \prec_P \dots \prec_P V(D_i^{m})$.
        \item We have $V(D_a^{1}) \prec_P  \dots \prec_P V(D_a^b)$.
        \item We have $V(\vargadget_1) \prec_P \dots \prec_P V(\vargadget_{a-1}) \prec_P V(\vargadget_a)$.
        \item For every $i < a$, $P$ traverses $D_{i}^{1},\dots,D_i^m$ in either true order or false order.
        \item $P$ traverses $D_{a}^{1},\dots,D_a^b$ in either true order or false order.
    \end{enumerate}
\end{definition}

\def\lastv{\mathsf{last}}

\begin{lemma}\label{lem:respect}
    Let $P$ be a Hamiltonian $(s, t)$-path of $G$.
    Then, $P$ is $(n,m)$-respectable.
\end{lemma}
\begin{proof}
    Throughout the proof, 
    we use the following observation that is immediate from the fact that 
    the vertices $\varvertexdt{in}{i}{j}$ and $\varvertexdt{out}{i}{j}$
    have degree two in $G$.
    \begin{observation}\label{obs:dt-vertices}
        For all $i \in [n]$ and $j \in [m]$, 
        the vertices 
        $\varvertex{0}{in}{i}{j}$,
        $\varvertexdt{in}{i}{j}$,
        $\varvertex{1}{in}{i}{j}$
        (resp., the vertices 
        $\varvertex{0}{out}{i}{j}$,
        $\varvertexdt{out}{i}{j}$,
        $\varvertex{1}{out}{i}{j}$%
        )
        appear consecutively in $P$.
    \end{observation}
    Assume towards a contradiction that $P$ is not $(n,m)$-respectable.
    Let $(a, b)$ be the 
    lexicographically maximum pair with $a\in [n]$ and $b\in[m]$ 
    such that $P$ is $(a,b)$-respectable. 
    That is, if $b<m$, then $P$ is not $(a,b+1)$-respectable, 
    and if $b=m$, then $P$ is not $(a+1,1)$-respectable.

    Note that we may assume that $P$ is $(1, 1)$-respectable 
    and therefore $(a, b) \ge (1, 1)$.
    $P$ starts in $s$, $s_1$, and then enters
    $D_1^1$. 
    By \Cref{obs:dt-vertices}, the next three vertices are 
    $\varvertex{1}{in}{1}{1}$,
    $\varvertexdt{in}{1}{1}$,
    $\varvertex{0}{in}{1}{1}$;
    either in that order or reversed.
    After that, $P$ may enter $\clausegadget_1$, 
    and then return to a vertex of $D_1^1$ by \Cref{lem:clause-gadget}
    which only leaves the choice of traversing 
    $\varvertex{0}{out}{1}{1}$,
    $\varvertexdt{out}{1}{1}$,
    $\varvertex{1}{out}{1}{1}$
    next, 
    again by \Cref{obs:dt-vertices}.
    Either the order is as shown in both cases (true-order),
    or the order is reversed in both cases (false-order).
    Therefore, $P$ is $(1, 1)$-respectable.

    Since $P$ is $(a,b)$-respectable, all the vertices of $D_a^b$ appear in $P$ and the vertex $x$ of $D_a^b$ that appears last in $P$ is either $\varvertex{1}{out}{a}{b}$ or $\varvertex{0}{out}{a}{b}$. By symmetry, we may assume that $x = \varvertex{1}{out}{a}{b}$.

    First, suppose that $b=m$. 
    Since $t_a$ is the only neighbor of $x$ that is not in $D_a^m$, we conclude that $t_a$ comes right after $\varvertex{1}{out}{a}{m}$ in $P$. 
    Now observe that the neighbors of $p_{a+1}$ are $t_1,\dots,t_a$ and $s_{a+1}$. Thanks to the $(a,b)$-respectability, we know that $t_1,\dots,t_{a-1}$ appear before $t_a$ in $P$. Consequently, the vertices $p_{a+1}$, $s_{a+1}$ must appear after $t_a$ in $P$, and the next vertex in $P$ must belong to $D_{a+1}^1$. 
    Now observe that the only neighbors of the vertices in $D_{a+1}^1$ are:
    
    \begin{enumerate}
        \item the vertices $\varvertex{1}{in}{a+1}{2}$ and $\varvertex{0}{in}{a+1}{2}$,
        \item potentially some vertices of the clause gadget $\clausegadget_1$, and 
        \item some vertices from the gadgets $\vargadget_1,\dots,\vargadget_a$.
    \end{enumerate}
    
    However, all the vertices from the gadgets $\vargadget_1,\dots,\vargadget_a$ appear in $P$ before $t_a$.
    From \Cref{lem:clause-gadget}, it follows that $P$ cannot ``escape'' $D_{a+1}^1$ by going to some vertices of $\clausegadget_1$, thus we deduce that $P$ traverse $D_{a+1}^1$ in true-order or false-order. Together with the fact that $P$ is $(a,m)$-respectable, we deduce that $P$ is also $(a+1,1)$-respectable. This is a contradiction with the maximality of $(a,b)$.
    
    From now, we suppose that $b<m$.
    Note that the only neighbors of $\varvertex{1}{out}{a}{b}$ are:
    
    \begin{enumerate}
        \item\label{enum:1out:1} from $D_a^b$,
        \item\label{enum:1out:2} from the gadgets $\vargadget_1,\dots,\vargadget_{a-1}$,
        \item\label{enum:1out:3} potentially from the clause gadget $\clausegadget_b$, and
        \item\label{enum:1out:4} the vertex $\varvertex{1}{in}{a}{b+1}$. 
    \end{enumerate}
    
    Since $P$ is $(a,b)$-respectable, the neighbors of $\varvertex{1}{out}{a}{b}$ from 
    \Cref{enum:1out:1,enum:1out:2} appear before $\varvertex{1}{out}{a}{b}$ in $P$.
    Moreover, the successor $y$ of $\varvertex{1}{out}{a}{b}$ in $P$ cannot be from the clause gadget $\clausegadget_b$ by \Cref{lem:clause-gadget} because $\varvertex{1}{in}{a}{b} \prec_P \varvertex{1}{out}{a}{b}$.
    Thus, $y = \varvertex{1}{in}{a}{b+1}$ and, by \Cref{obs:dt-vertices}, it is followed by $\varvertexdt{in}{a}{b+1}$ and $\varvertex{0}{in}{a}{b+1}$.

    \medskip
    
    We claim that $\varvertex{0}{in}{a}{b+1}$ and its successor in $P$ are adjacent through a dummy edge.  
    Indeed, if it is not the case, then the successor of $\varvertex{0}{in}{a}{b+1}$ in $P$ is either $\varvertex{0}{out}{a}{b+1}$ or a vertex $w$ from $\clausegadget_{b+1}$.
    In the latter case, by \Cref{lem:clause-gadget} $w$ is followed in $P$ by all the vertices from $\clausegadget_{b+1}$ and then $\varvertex{0}{out}{a}{b+1}$.
    In both case, the successors of $\varvertex{0}{out}{a}{b+1}$ in $P$ must be $\varvertexdt{out}{a}{b+1}$ and $\varvertex{1}{out}{a}{b+1}$ by \Cref{obs:dt-vertices}. But, then $P$ would traverse $D_a^1,\dots, D_a^b,D_a^{b+1}$ in true-order which implies that $P$ is $(a,b+1)$-respectable: a contradiction with the maximality of $(a,b)$.
    
    \smallskip
    
    Thus, $\varvertex{0}{in}{a}{b+1}$ and its successor $w$ in $P$ are adjacent through a dummy edge.
    In particular, we know that $w$ is different from $\varvertex{0}{out}{a}{b+1}$.
    See \Cref{fig:lem-respect} for an illustration.
    The neighbors of $\varvertex{0}{out}{a}{b+1}$ are: 
    
    \begin{enumerate}
        \item\label{enum:0out:1} $\varvertex{0}{in}{a}{b+1}$, 
        \item\label{enum:0out:2} some vertices from the gadgets $\vargadget_1,\dots,\vargadget_{a-1}$,
        \item\label{enum:0out:3} potentially a vertex from  $\clausegadget_{b+1}$, 
        \item\label{enum:0out:4} $\varvertexdt{out}{a}{b+1}$, and 
        \item\label{enum:0out:5} $\varvertex{0}{in}{a}{b+2}$ if $b+1 < m$, or $t_a$ if $b+1=m$.
    \end{enumerate}

    We know that $\varvertex{0}{in}{a}{b+1}$ has two neighbors in $P$ distinct from $\varvertex{0}{out}{a}{b+1}$.
    Above, we argued that 
    $\varvertex{0}{in}{a}{b+1}$ 
    and 
    $\varvertex{0}{out}{a}{b+1}$
    are not adjacent in $P$,
    ruling out \Cref{enum:0out:1}.
    Since $P$ is $(a,b)$-respectable, we know that $ V(\vargadget_1) \prec_P \dots \prec_P V(\vargadget_{a-1}) \prec_P V(\vargadget_a)$, thus the vertices from \Cref{enum:0out:2} appear before $\varvertex{0}{in}{a}{b+1}$ in $P$, so they cannot be adjacent to $\varvertex{0}{out}{a}{b+1}$ in $P$.
    From \Cref{lem:clause-gadget}, we deduce that $\varvertex{0}{out}{a}{b+1}$ is not adjacent in $P$ to its potential neighbor in $\clausegadget_{b+1}$ because $\varvertex{0}{in}{a}{b+1}$ is not adjacent in $P$ to its potential neighbor in $\clausegadget_{b+1}$,
    ruling out \Cref{enum:0out:3}.
    Hence, the neighbors of $\varvertex{0}{out}{a}{b+1}$ in $P$ are $\varvertexdt{out}{a}{b+1}$ and $\varvertex{0}{in}{a}{b+2}$ if $b+1 < m$, or $t_a$ if $b+1=m$.
    Symmetrically, the neighbors of $\varvertex{1}{out}{a}{b+1}$ in $P$ are $\varvertexdt{out}{a}{b+1}$ and $\varvertex{1}{in}{a}{b+2}$ if $b+1 < m$, or $t_a$ if $b+1=m$.
    By \Cref{obs:dt-vertices}, $\varvertex{1}{out}{a}{b+1}$, $\varvertexdt{out}{a}{b+1}$ and $\varvertex{0}{out}{a}{b+1}$ must appear consecutively in $P$.
    
    If $b+1=m$, then $\varvertex{1}{out}{a}{b+1}$ and $\varvertex{0}{out}{a}{b+1}$ are both adjacent to $t_a$ in $P$, and $P$ would contain a cycle of size four: a contradiction.
    Now, if $b+1 <m$, then $\varvertex{1}{out}{a}{b+1}$ and $\varvertex{0}{out}{a}{b+1}$ are adjacent in $P$ to $\varvertex{1}{in}{a}{b+2}$ and $\varvertex{0}{in}{a}{b+2}$, respectively. But by \Cref{obs:dt-vertices}, these two latter vertices must be adjacent in $P$ to $\varvertexdt{in}{a}{b+2}$. This implies that $P$ contains a cycle of size six (see \Cref{fig:lem-respect}): a contradiction.
    \begin{figure}[hbt]
        \centering
        \includegraphics[width=.8\textwidth]{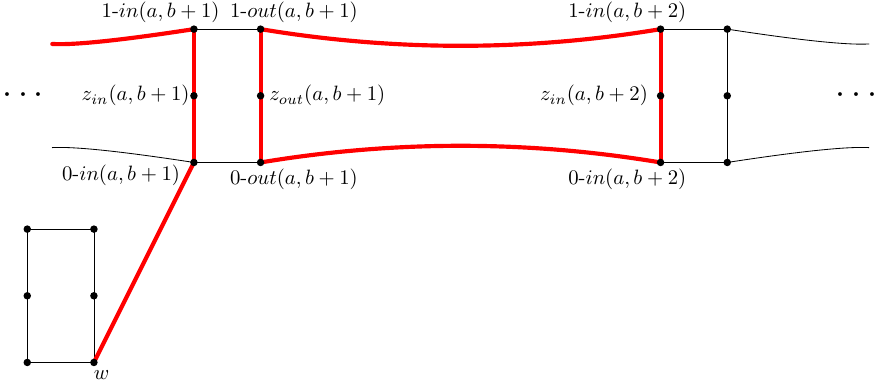}
        \caption{The situation when a Hamiltonian path uses a dummy edge.}
        \label{fig:lem-respect}
    \end{figure}
\end{proof}

\begin{lemma}\label{lem:cor-path-to-ass}
    If $G$ has a Hamiltonian path, then $\cnf$ is satisfiable.
\end{lemma}
\begin{proof}
    Let $P$ be a Hamiltonian path in $G$.
    By \Cref{lem:respect}, $P$ is $(n,m)$-respectable and in particular, for every $i\in[n]$, $P$ traverses $D_i^1,\dots,D_i^m$ in either true-order or false-order.
    We construct a truth assignment $\ass \colon \var(\cnf) \to \{0, 1\}$ as follows.
    For each $i \in [n]$, if $P$ traverses $D_i^1,\dots,D_i^m$ in true-order, 
    we set $\ass(x_i) = 1$, 
    and if $P$ traverses $D_i^1,\dots,D_i^m$ in false-order,
    we set $\ass(x_i) = 0$.
    Let $j \in [m]$.
    Since $P$ is a Hamiltonian path, it collects the vertices in $\clausegadget_j$.
    By \Cref{lem:clause-gadget}, there is some $i \in [n]$ such that 
    $P$ contains a path starting with a vertex in $D_i^j$,
    traversing all vertices in $\clausegadget_j$, and ending with a vertex in $D_i^j$.
    By construction, the two vertices in $D_i^j$ are
    $\varvertex{0}{in}{i}{j}$, $\varvertex{0}{out}{i}{j}$,
    if $x_i$ appears positively in $C_j$, and
    $\varvertex{1}{in}{i}{j}$, $\varvertex{1}{out}{i}{j}$,
    if $x_i$ appears negated in $C_j$.
    Suppose without loss of generality that the former is the case.
    For $P$ to contain the described subpath,
    it must have entered $D_i^j$ at vertex
    $\varvertex{1}{in}{i}{j}$,
    which implies that $P$ traverses $D_i^1,\dots,D_i^m$ in true-order.
    By our construction, $\ass(x_i) = 1$, meaning that the literal of $C_j$ containing $x_i$
    satisfies $C_j$.
\end{proof}

\begin{lemma}
    We can construct in polynomial time a linear order on $V(G)$ of mim-width at most $25$.
\end{lemma}
\begin{proof}
    Let $\linord$ be any linear order on $V(G)$ extending the following partial order:
    \begin{align*}
        s, s_1, &D_1^1, p_2, s_2, D_2^1, \ldots, p_n, s_n, D_n^1, V(\clausegadget_1), \ldots, \\
        &D_1^2, D_2^2, \ldots, D_n^2, V(\clausegadget_2), \\
        &\ldots, \\
        &D_1^{m-1}, D_2^{m-1}, \ldots, D_n^{m-1}, V(\clausegadget_{m-1}), \\
        &D_1^{m}, t_1, D_2^{m}, t_2, \ldots, D_n^{m}, t_n, t, V(\clausegadget_{m}) 
    \end{align*}

    Throughout the following, we let $D = \bigcup_{(i,j) \in [n] \times [m]} D_i^j$,
    $P = \{p_i \mid i \in [n]\}$, $S = \{s\} \cup \{s_i\mid i \in [n]\}$,
    $T = \{t_i \mid i \in [n]\} \cup \{t\}$.
    Let $(A, \bar{A})$ be any cut induced by $\linord$ (such that $A$ contains the first $t$ vertices of $\linord$ for some $t\leq |V(G)|$),
    and let $M$ be an induced matching in $G[A, \bar{A}]$.
    Each edge of $M$ is one of the following types, 
    and below we justify the bounds on their number indicated in the respective parenthesis.
    
    \begin{enumerate}
        \item\label{enum:mim-SPT} 
            Between $S \cup P$ and $S \cup T$ ($\le 1$).
        \item\label{enum:mim-STD} 
            Between $S \cup T$ and $D$ ($\le 1$).
        \item\label{enum:mim-DC} 
            Between $D$ and $V(\clausegadget_j)$, for some $j \in [m]$ ($\le 6$).
        \item\label{enum:mim-C} 
            Inside $V(\clausegadget_j)$ for some $j \in [m]$ ($\le 13$).
        \item\label{enum:mim-D} 
            Inside $D$ ($\le 4$).
    \end{enumerate}
    
    For edges of type \ref{enum:mim-SPT},
    note that $G[A \cap (S \cup P), \bar{A} \cup (S \cup T)]$ is a chain graph plus 
    possibly one edge $p_is_i$, for some $i \in [n]$.
    Recall by \Cref{obs:chain}, $M$ can contain at most one edge from any chain graph.
    Furthermore, if $p_is_i \in M$, then $M$ has no other edge from this subgraph, 
    since in this case, $p_i \in A$
    is adjacent to all vertices in $T \cap \bar{A}$ 
    that have a neighbor in $(S \cup P) \cap A$.
    So in either case, there is at most one edge of type \ref{enum:mim-SPT} in $M$.
    
    For edges of type \ref{enum:mim-STD}, such edges can only occur if either some $s_i$
    is the last vertex of $A$ or if some $t_i$ is the first vertex of $\bar{A}$;
    in any case, there is at most one such edge in $M$.
    The number of edges of types \ref{enum:mim-DC} and \ref{enum:mim-C} are trivial upper bounds.
    
    Lastly, consider edges of type \ref{enum:mim-D}.
    Let $(j, i)$ be the lexicographically largest pair such that $A$ contains vertices from $D_i^j$.
    First, there are at most two edges from $D_i^j$ itself in $M$.
    Next, consider the edges in $M$ that are between 
    $D^j = \bigcup_{i \in [n]} D_i^j$ and
    $D^{j+1} = \bigcup_{i \in [n]} D_i^{j+1}$.
    Suppose there is some $h \in [n]$ such that there is an edge 
    $d_j d_{j+1}$ in $M$, where $d_j \in D_h^j$ and $d_{j+1} \in D_h^{j+1}$.
    Note that there are at most two such edges per $h$.
    Moreover, in this case, $M$ has no other edges between $D^j$ and $D^{j+1}$,
    since for each such edge, at least one endpoint has a neighbor in $\{d_j, d_{j+1}\}$.
    So in this case, there are at most four edges of type \ref{enum:mim-D} in total.
    If there is no such $h$, then by a similar argument, 
    $M$ contains at most one edge between $D^j$ and $D^{j+1}$,
    in which case we have at most three edges of type \ref{enum:mim-D}.
\end{proof}

This finishes the proof for \textsc{Hamiltonian Path}.
By adding another degree-two vertex adjacent to only $s$ and $t$, 
we also obtain the same result for the \textsc{Hamiltonian Cycle} problem and the fact that adding a vertex to a graph increase its mim-width by at most~$1$.
\thmMain*